\documentclass[11pt]{amsart}

\usepackage[utf8]{inputenc}
\usepackage[T1]{fontenc}   
\usepackage[english]{babel} 
\usepackage{fullpage}
\usepackage{subfig}

\usepackage{amssymb,libertine}
  
\usepackage{array}
\usepackage{bbm}

\usepackage{mathtools}
\usepackage[svgnames]{xcolor}
\usepackage{hyperref}

\usepackage{wrapfig}
\usepackage{tikz}
\usepackage{color}
\usepackage{graphicx}
\usepackage{amssymb}
\usepackage{epstopdf}
 \usepackage{latexsym}
\usepackage{amsmath}
\usepackage{amsfonts} 
\usepackage{amsthm}
\usepackage[export]{adjustbox}

   \def\CC{\mathbb{C}}
    \def\DD{\mathbb{D}}
    \def\NN{\mathbb{N}}
    
    \def\RR{\mathbb{R}}
    \def\ZZ{\mathbb{Z}}
  
    \newtheorem{Proposition}{Proposition}
\newtheorem{Theorem}[Proposition]{Theorem}
\newtheorem{Lemma}[Proposition]{Lemma}
\newtheorem{Definition}[Proposition]{Definition}
\newtheorem{Corollary}[Proposition]{Corollary}
\newtheorem{Remark}[Proposition]{Remark}
\newtheorem{Note}[Proposition]{Note}
\newtheorem{Notations}[Proposition]{Notations}
\def\cFrac#1#2{%
\begin{array}{@{}c@{}}\multicolumn{1}{c|}{#1}\\%
\hline\multicolumn{1}{|c}{#2}\end{array}}

\def\z{\noindent}  
\def\be{\begin{equation}}
\def\ee{\end{equation}}
    \def\la{\langle}
\def\ge{\geqslant}
\def\le{\leqslant}
\def\ra{\rangle}

\def\bd{\begin{Definition}}
\def\ed{\end{Definition}}
\def\bt{\begin{Theorem}}
\def\et{\end{Theorem}}

\def\epsilon{\varepsilon}
\def\bel{\begin{equation}\label}
\def\ee{\end{equation}}

\def\Ei\text{Ei}

\def\phi{\varphi}
\def\bfu{\mathbf{u}}

 \def\la{\langle}
\def\ra{\rangle}
\def\q{q}

\def\a{\alpha}
\def\gg{  \mathbf{g}}

\def\ff{  \mathbf{f}}
\def\yy{  \mathbf{y}}
\def\ww{  \mathbf{w}}

\def\HHL{\mathbb{H_\ell}}

\def\om{\omega}

\def\Le{\leqslant}
\def\Ge{\geqslant}

\def\f{\phi}
\def\F{\Phi}
\def\et{\eta}

\title{Nonperturbative time dependent solution of a simple ionization model.}
\author{Ovidiu Costin}\address{The Ohio State University, 231 W 18th Ave, Columbus, OH 43210}\email{costin.9@osu.edu}
\author{Rodica D. Costin}\address{The Ohio State University, 231 W 18th Ave, Columbus, OH 43210} \email{costin.10@osu.edu}
\author{Joel L. Lebowitz}\address{Departments of Mathematics and Physics, Rutgers University, Hill Center - Busch Campus, 
110 Frelinghuysen Road
Piscataway, NJ 08854}\email{ lebowitz@math.rutgers.edu}

\begin{document}
\maketitle


\begin{abstract}
 We present a non-perturbative solution of the Schr\"odinger equation $i\psi_t(t,x)=-\psi_{xx}(t,x)-2(1 +\a \sin\omega t) \delta(x)\psi(t,x)$, written  in units in which $\hbar=2m=1$,  describing the ionization of a model atom by a parametric oscillating potential. This model has been studied extensively by many authors, including us. It has surprisingly many features in common with those observed in the ionization of real atoms and emission by solids, subjected to microwave or laser radiation. Here we use new mathematical methods to go beyond previous investigations and to provide a complete and rigorous analysis  of this system. We obtain the Borel-resummed transseries (multi-instanton expansion)  valid for all values of $\a,\omega,t$ for the wave function, ionization probability, and   energy distribution of the emitted electrons, the latter not studied previously for this model.  We show that for large $t$ and small $\a$ the energy distribution has sharp peaks at energies which are multiples of $\omega$, corresponding to photon capture. We obtain small $\a$ expansions that converge for all $t$, unlike those of standard perturbation theory.
 We expect that our analysis will serve as a basis for treating more realistic systems revealing a form of  universality in different emission processes. 
 \end{abstract}

\section{Introduction}

The ionization of atoms and the emission of electrons from a metal, induced by an oscillating field, such as one produced by a laser, continues to be a problem of great theoretical and practical interest, see \cite{BauerD}, \cite{CT2}, \cite{28}, \cite{Zhang} and the references therein. This phenomena goes under the name of photo-emission.
It was first explained by Einstein in 1905; an electron absorbs "$n$ photons" acquiring  their  energy, $n\hbar \omega$, which permits it to escape the potential barrier confining it. While the complete physics of these phenomena would involve quantization of the electromagnetic field and its interaction with matter, i.e. photons and relativity, the basic understanding is contained already in the semiclassical limit where the electromagnetic field is not quantized, expected to be valid when the density of photons is large \cite{Intro}; for a mathematical derivation of this limit via Floquet states see \cite{Floquet}. One then considers the solution of the non-relativistic Schr\"odinger equation in an oscillating field giving rise to a potential with period $2\pi/\omega$, \cite{BauerD}, \cite{CT2}. Resonant energy absorption at multiples of $\omega$ then yields effects qualitatively similar to those of photons, in some regimes, see Fig.\,\ref{fig12}.

In units in which $\hbar=2m=1$ the Schr\"odinger equation has the form
\begin{equation}\label{Scho}
i\frac{\partial \psi}{\partial t}=\left[ H_0+V(t,x)\right]\psi
\end{equation}
Here $H_0$ describes the time-independent system assumed to have  both discrete and continuous spectrum, and  the laser field is modeled by a time periodic potential,  $V(t,x)=V(t+2\pi/\om,x)$. Typically, the latter is represented as a vector potential or a dipole field, e.g. $V(t,x)=E\cdot x\, \sin\om t$, \cite{BauerD}, \cite{CT2}.

Starting in a bound state of the reference hamiltonian $H_0$, $\psi(x,0)=u_b(x)$ corresponding to the  energy $-E_b$ and expanding in generalized eigenstates, assuming $u_b$ is the only effective bound state, the evolution is given by
\begin{equation}
  \label{eq:schr}
\psi(x,t)=\theta(t)e^{iE_bt}u_b(x)+\int_{\RR^d} \Theta(k,t)u(k,x)e^{-ik^2t}\, dk
\end{equation}
Physically, $|\theta|^2$ gives the probability  of finding the particle in the eigenstate $u_b(x)$ and  $|\Theta(k,t)|^2$ is the  probability density of the ionized electron in  ``quasi-free'' states (continuous spectrum) with energies $k^2$.
It follows from the unitarity of the evolution that
$$ |\theta(t)|^2+\int_{\RR^d} |\Theta(k,t)|^2\, dk=1$$
 Accordingly, if $\theta(t)\to 0$ as $t\to\infty$, we say that the system ionizes completely.

 When $\om>E_b$,  a first order approximation \cite{BauerD}, \cite{CT2} in the strength of $V$ (used very judiciously) gives  emission into states with energy $k^2+E_b=\om$. 
Clever physics arguments also yield Fermi's golden rule of exponential decay from the initial bound state \cite{BauerD}, \cite{CT2}. These only hold approximately and only over some ``intermediate'' time scales as discussed in the sequel.

To deal with the case of transitions caused by large fields $V$ one needs to go to high order perturbation theory, which is complicated \cite{BauerD}, \cite{CT2}. In fact, as we will explain, standard perturbation theory only produces a finite number of correct perturbative orders. To deal with larger fields  one uses various "strong field" approximations due to Keldysh and others \cite{Keldysh}. For literature on strong field approximations see \cite{tutorial}, \cite{22}, \cite{37}, \cite{Popruzhenko}. There, one uses scattering states $\tilde{u}(k,x)$ strongly modified (Volkov states) by the oscillating field.
We shall not consider that here but focus on getting a complete rigorous solution of \eqref{Scho} for a toy model which nevertheless exhibits many features of more realistic situations, see \cite{OCJLAR}. We can then study carefully how photons show up in this semiclassical limit.

The model we study is a one dimensional system with reference Hamiltonian $H_0$, whose mathematical properties are analyzed in \cite{cycon},  is
\begin{equation}
  \label{eq:h0}
H_0=-\frac{\partial^2}{\partial x^2}-2\delta(x) ,\ \ x\in\RR,
\end{equation}
It  has a single bound state 
$$u_b(x)=e^{-|x|}$$ with energy $-E_b=-1$ and its generalized eigenfunctions are
\begin{equation}\label{four}
u(k,x)=\frac{1}{\sqrt{2\pi}} \left( e^{ikx}-\frac{ e^{i|kx|}}{1+ik}  \right),\ \ x,k\in\RR
\end{equation}
Beginning at $t=0$, when $\psi(x,0)=u_b(x)$, we add a parametric harmonic perturbation to the base potential. For $t\Ge 0$ we have
\begin{equation}
  \label{eq:field}
 H=-\frac{\partial^2}{\partial x^2}-2\delta(x) -2\a \sin\omega t\, \delta(x)=H_0+V(t,x)
\end{equation}
 (where we take for definiteness $\a,\omega>0$)
and look for solutions of the associated Schr\"odinger equation in the form \eqref{eq:schr}.
The full behavior of $\psi(x,t)$ is very complicated despite the simplicity of the model. We expect the main feature of the evolution of $\psi(x,t)$ to be universal for ionization by an oscillatory field. 

As already noted this model has been studied extensively before. We refer the reader in particular to \cite{OCJLAR} where it was shown that, for all $\alpha$ and $\omega$,  $\theta(t)\to 0$, i.e., we have complete ionization. We also investigated there both analytically and numerically the behavior of $\theta(t)$ as a function of $\omega$ and showed qualitative agreement with experiments on the ionization of hydrogen-like atoms by strong radio frequency fields. In \cite{CMP} we studied general periodic potentials and found the condition on the Fourier coefficients for complete ionization. There are (exceptional) situations where one does not get complete ionization. In \cite{CLS} we showed ionization when the external forcing is an oscillating electric field. A large field approximation for this latter setting can be found in \cite{Elberfeld}.

In this paper we introduce new methods which allow us to complete the analysis of this model for all $t,\alpha$: we obtain a rapidly convergent representation (in the form of a Borel summed transseries, or ``multi-instanton expansion'')  for the solution $\psi(x,t)$ valid for all $t,\,\omega$ and $\a$ and  we find the distribution $|\Theta(k,t)|^2$ of energies of the emitted electrons as a function of $t,\a,\omega$. The latter, which was not done before, is where the  "photonic" picture shows up most clearly. We will investigate this connection more explicitly in a separate article \cite{prepa}. 

There are strong peaks  of $|\Theta(k,t)|^2$ which for small $\alpha$ and $\omega\in(\frac 1n, \frac 1{n-1})$ are centered  near
$k^2= n\omega-1$, see Fig. 1 for $\omega=3/2$. The main peak corresponds to the absorption of one photon and approaches a Dirac distribution centered at $k^2=1/2$ in the limit $t\to\infty$ followed by $\a \to 0$. Clearly,  the discreteness of the emission spectrum in the above limit is a consequence of the
periodicity of the classical oscillating field and does not require the 
concept of photons, see also  \cite{Popruzhenko}, Footnote 1.
We find that   there are other (smaller) peaks emanating from the  bottom of the continuous spectrum. For small $\alpha$ these are centered near $k^2=n\omega$, see Theorem \ref{T4}, (iv). 
We also obtain a perturbation expansion of the wave function for small $\alpha$ in a form which is uniformly convergent for any $t\in\RR^+$, and which, in principle, can be carried out explicitly to any order. 

It follows from our analysis that the predictions of the usual perturbation theory hold when  $t=o(\a ^{-2}|\ln\a|)$, beyond which the behavior of the physical quantities is qualitatively different.

\subsection{The Laplace transform and the energy representation}\label{prior}  It was shown in \cite{OCJLAR} that

\begin{equation}
  \label{eq:theta1}
  \theta(t)=1+2i\int_0^t \f(s)ds
\end{equation}
and 
\begin{equation}
  \label{eq:theta1T}
 \Theta(k,t)=\sqrt{\frac{2}{\pi}}\frac{|k|}{1-i|k|}\,\int_0^t\f(s)\, \mathrm{e}^{i(1+k^2)s}\, ds
\end{equation}
where $\f$ satisfies the integral equation
$$\f(t)=\a\sin\omega t\left(1+\int_0^t \f(s)\et(t-s)ds\right)$$
with 
$$\et(s)=\frac{2i}{\pi}\,\int_0^\infty\,\frac{u^2e^{-is(1+u^2)}}{1+u^2}\, du=\frac{\sqrt{i} e^{-i s}}{\sqrt{\pi } \sqrt{s}}-i \text{erfc}\left(\sqrt{i s}\right) $$

It can be checked, \cite{OCJLAR},  that the Laplace transform of $\phi$
\begin{equation}\label{star}
\F(p):=\mathcal{L}\phi(p)=\int_0^{\infty}\f(s)e^{-ps}ds
\end{equation}
is analytic in the right half plane and satisfies the functional equation\footnote{A very similar functional equation can be obtained directly from the Schr\"odinger equation for $\mathcal{L}\psi(0,p)$.}
\begin{equation}
  \label{eq7}
  \F(p)=\frac{i\a }2\frac{\F(p-i\omega)}{i\sqrt{ip+\omega-1}+1} -\frac{i\a }2\frac{\F(p+i\omega)}{i\sqrt{ip-\omega-1}+1} +\frac{\a\omega}{\omega^2+p^2}
\end{equation}
(The square root is understood to be positive on $\RR^+$, and analytically continued on its Riemann surface. \footnote{In previous papers we used $r$ instead of $\a$ and a different branch of the square root; with these changes the formulas agree.})

\section{Main results}
\subsection{Results for general $\a,\omega$}

\begin{Theorem}\label{T1}
  For all $\a>0$ and $\omega>0$,

(i)  $(1+|p|)^2\F(p)$ is bounded and $\F$ is analytic in the closed right half plane, except for
  \begin{equation}\label{valbn}
p=-i \beta_n,\ \ \   \beta_n=1+n\omega, \ \ \ \ n\in\ZZ
  \end{equation}
   where it is
  analytic in $\sqrt{p+i \beta_n}  $; 

  (ii) in the open left half plane $\F$ has exactly one array of simple poles located at
  \begin{equation}
    \label{eq:42}
    p_n:=p_n(\a,\omega)=p_0(\a,\omega) -in\omega,\ \  \ \ n\in \ZZ
  \end{equation}
 and the residues $R_n=\text{res}(\Phi,p_n)$ can be calculated using continued fractions, see \S\ref{PropHomEq}, and satisfy
  \begin{equation}
    \label{eq:43}
    R_n=O(|n!|^{-1/2});\ \ \ |n|\to\infty
  \end{equation}
   Away from the line of poles, $(1+|p|)^2\F(p)$ is bounded in the left half plane. The functions $p_n$ and $R_n$ are analytic in $\a$ in a neighborhood of $[0,\infty)$;

  (iii) $\Re p_0(\a;\omega)<0$;  $p_0$ satisfies an equation of the form $A(p,\a;\omega)=B(p,\a;\omega)$ where $A,B$ are meromorphic functions (given by convergent continued fraction representations, see \eqref{eq:24},\,\eqref{eq:24O}).
  \end{Theorem}
  
  \begin{figure}[h!]\hspace{-2.cm}
    \subfloat[$t\le 150$]{\includegraphics[scale=0.5]{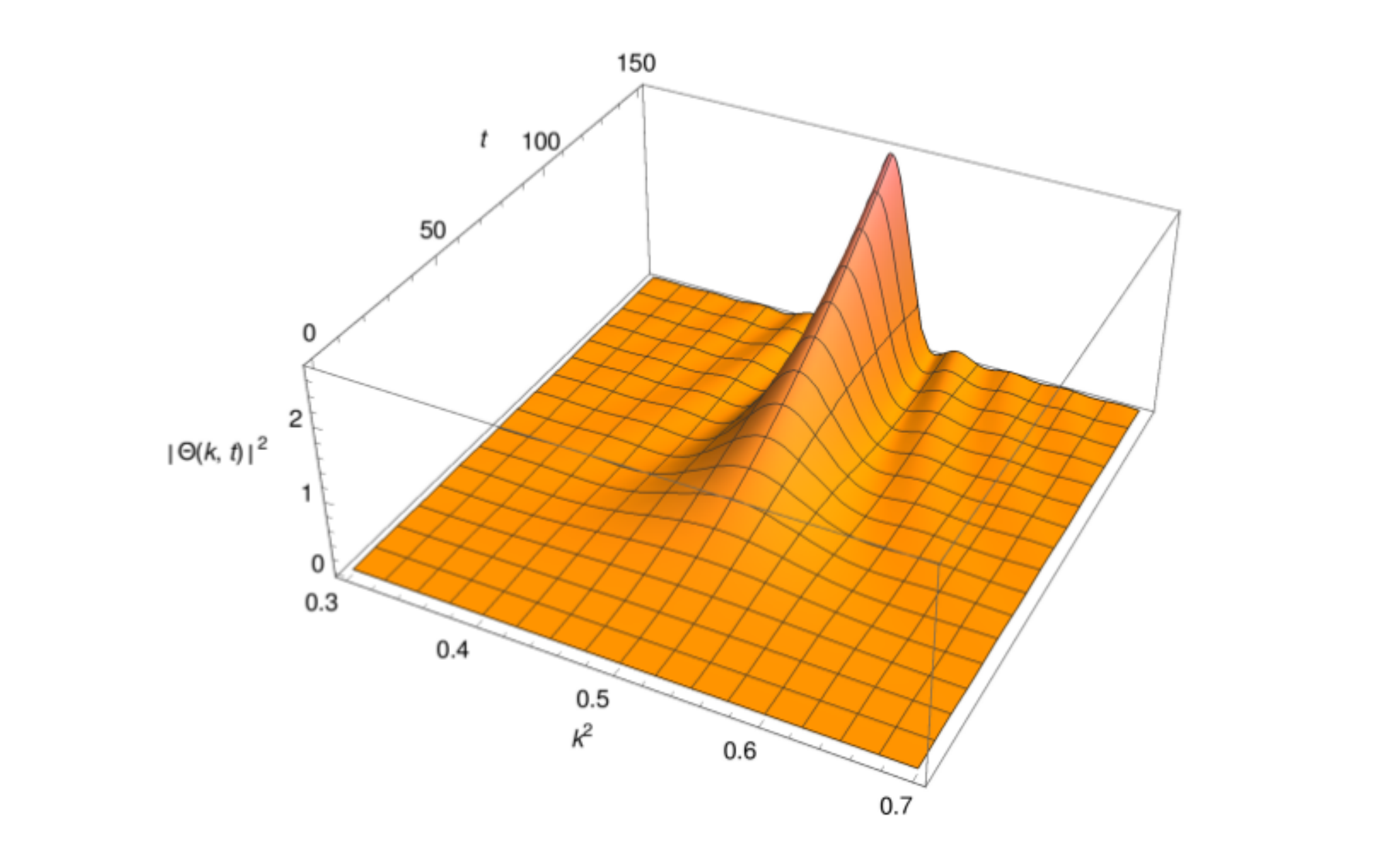}}\hspace{-1.5cm}
      \raisebox{0.3cm}{ \subfloat[$500\le t\le 1500$]{\includegraphics[scale=0.4]{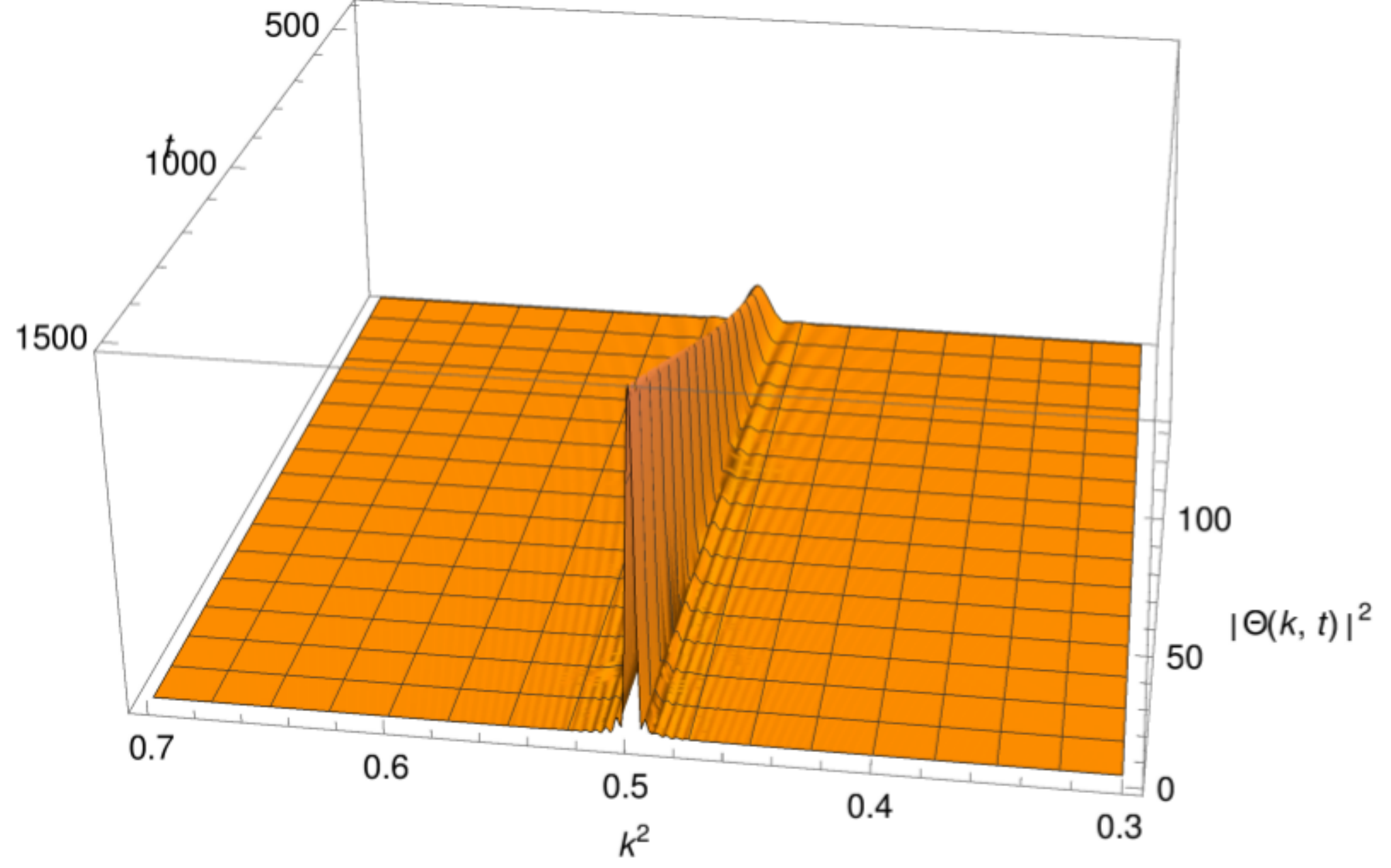}}}
    \caption{ $\Theta$ as a function of $k^2$ and time, for $\alpha=1/20$ and $\omega=3/2$, calculated from the leading order in \eqref{eq:ThetaSmall a}.}
\label{fig12}
\end{figure}

The following is the \emph{non-perturbative} (arbitrary coupling) form of the decay of the bound state.
\begin{Theorem}\label{Theo2}

  (i) The function $\theta$ in \eqref{eq:theta1} has a Borel summed {\bf transseries} representation (also known as a {\bf multi-instanton} expansion, \cite{Zinn}) convergent for all  $t>0$ 
    \begin{equation}
    \label{eq:19}
    \theta(t)=2 i\sum_{n\in\ZZ}\frac{R_n}{p_n}e^{p_n t}+\sum_{n\in\ZZ} e^{-i  \beta_n t}\int_0^{\infty}e^{-st} F_n(s;\a)ds
  \end{equation}
  with $ \beta_n\in\RR,\ p_n\in\CC,\,\Re p_n<0,\ R_n\in\CC$ as  in \eqref{valbn}-\eqref{eq:43}. The $F_n$ are analytic in $\a$ in a neighborhood of $[0,\infty)$, analytic in $s$ if $\Re s>0$ and in $\sqrt{s}$ for $s$ near $0$.  For large $|s|$, $F_n=O(s^{-3})$. The first sum converges factorially and the second at least as fast as $1/n^3$.

  (ii) Similarly, the function $\Theta$ is a Borel summed transseries 
  \begin{equation}
    \label{eq:341}
\Theta(k,t)=\sqrt{\frac{2}{\pi}}\frac{|k|}{1-i|k|}\left[\Phi(-i(1+k^2))+\sum_{n\in\ZZ}\frac{R_n e^{(p_n+ i+ik^2)t}}{p_n+i(1+k^2)}+\sum_{n\in\ZZ} e^{-i \beta_n t}\int_0^{\infty}e^{-st} G_n(s;\a)ds\right]
\end{equation}
where the $G_n$ have the same properties as the $F_n$.
\end{Theorem}
\begin{Corollary}
 For all $\a,\omega >0$ we have
  \begin{equation}
    \label{eq:34}
\lim_{t\to\infty}\theta(t)=0,\ \  \lim_{t\to\infty} \Theta(k,t)=\sqrt{\frac{2}{\pi}}\ \frac{|k|}{1-i|k|}\ \Phi(-i(1+k^2))
\end{equation}
  
\end{Corollary}

\subsection{Perturbation theory: results for small $\a$.}\label{PerturbTh}

In this section we assume that $\omega^{-1}\notin \NN$. See Note\,\ref{Note5}  regarding $\omega^{-1}\in \NN$. 

{\bf Notation:} In the rest of the paper ``$o_a$'' denote  functions analytic and vanishing at $\a=0$.

\begin{Theorem}\label{T4} 

Assume $\omega^{-1}\notin \NN$. Let $p_0=p_0(\a,\omega)$ as in Theorem\,\ref{T1}\, (ii).

(i)  For $\omega>1$, we have, for $\a$ small enough,
  \begin{equation}
    \label{eq:38}
    p_0=-\a^2\,\frac {\sqrt {\omega-1}+i\sqrt {1+\omega}}{2\omega}\, (1+o_a)
  \end{equation}

With $m$ the least integer for which $m\omega>1$ we have   
  \begin{equation}
    \label{eq:39}
    \Re p_0=-\frac{1}{m\omega}\frac{\sqrt{m\omega-1}}{\prod_{k<m}(1-\sqrt{1-k\omega})^2}\frac{\a^{2m}}{2^{2m+1}}(1+o_a)
  \end{equation}
  
  (ii) The residues (see \eqref{eq:43} for arbitrary $\alpha$) satisfy
\begin{equation}
    \label{eq:333}
R_{0} = \frac{i m \a^m p_0}{2^m \prod_{k<m}(1-\sqrt{1-k\omega})} (1+o_a) 
\end{equation}
 Furthermore, as $n\to\infty$, $$R_n/R_0=O(\a ^{2|n|}/\Gamma(n/2))\ \  \text{and} \  G_n=O(\a^{2n+2}/n^3)$$
where the $G_n$ are defined in \eqref{eq:341}.

 (iii) As a function of $\a$,  $\theta(t)$ is  analytic for  small  $\a\in \CC$ and real-analytic for $\a\in\RR$.
 
 With $m$ as in (i),  on the scale $t\in (0,o(\a ^{-2m}\ln\a)) $ we have\footnote{This is the Fermi Golden Rule of exponential decay of the bound state valid for small amplitudes and moderately large times.} 
 $$|\theta(t)|^2= e^{-2\Re p_0 t}(1+o(1))$$
As $t\to\infty$,
$$|\theta(t)|^2=O(\a ^4 t^{-3})$$

 (iv) The distribution of energies satisfies
 \begin{equation}
   \label{eq:ThetaSmall a}
   \Theta(k,t)=\sqrt{\frac2{\pi}}\,\,\frac{\a \omega k}{1-ik}\,\,\frac{1-e^{-\frac{\a  ^2 t \sqrt{\omega -1}}{2 \omega }} e^{i t \left(\frac{\a  ^2
   \sqrt{\omega +1}}{2 \omega }+k^2-\omega +1\right)}}{\a  ^2 \left(\sqrt{\omega -1}-i
   \sqrt{\omega +1}\right)-2 i \omega  \left(k^2-\omega +1\right)}(1+o_a)
\end{equation}

\end{Theorem}
\begin{Note}\label{Note5}
  If for some $m\in\NN$ we have $\ m\omega^{-1}=1+O(\a^{2m})$, which means that poles are close to branch points,  there is a smooth transition region where $R_0$ and $\Re p_0$ change from $O(\a^{2m+2})$ to $O(\a^{2m})$. We will not analyze this  intricate transition in the present paper.
\end{Note}

\begin{Corollary}
 The functions $\Theta$ and $\theta$ have fully convergent perturbation expansions in small $\a$, in sup norm away from $t=0$, provided we keep the power series of $p_n$ in the exponent in \eqref{eq:19} and \eqref{eq:341}.
\end{Corollary}
\begin{proof}
  This follows from uniform convergence and analyticity of each term in \eqref{eq:19} and \eqref{eq:341}. For the expansion to be uniformly rapidly convergent the exponentials should not be further expanded as  power series.
\end{proof}
\section{Proofs}
\subsection{Organization of the paper and main ideas}  We are interested in obtaining rapidly convergent expansions for  $\theta(t)$ and $\Theta(k,t)$ for all $\a$ and $\omega$. To achieve this we study in great detail the singularity structure of $\Phi(p)$. We prove in particular that $\Phi(p)$ has exactly one array of evenly spaced poles, for $\Re p<0$, and one array of branch points, for $\Re p=0$. Their location and residues determine, via the inverse Laplace transform, the transseries representation of $\Theta$ and $\theta$. To  show this rigorously for all $\a$ we first establish these facts for small $\a$ using compact operator techniques; we then extend them for arbitrary $\a$ by devising a periodic operator isospectral with the one of interest, whose pole structure can be analyzed by appropriate complex analysis tools.

The proof of Theorem\,\ref{T1} (i) is found in \S\ref{pfT1}, (ii) and (iii) in \S\ref{Sec5}, \eqref{eq:43} in \S\ref{abstract}. The functional equation \eqref{eq7} is rewritten as a parameter dependent equation on $\ell^2(\ZZ)$ and analyzed with compact operator techniques.

Section \S\ref{PropHomEq} contains results and notations used further in the paper. 

Theorem\,\ref{T4} (i) is proved in \S\ref{polsmalr}, (ii) in \S\ref{cabstract} and (iii),  (iv) in \S\ref{PfT434}. For small $\a$ the position of the poles is found from a continued fraction representation described in \S\ref{PropHomEq}. The information about the poles for larger  $\a$ relies on the analysis of a periodic compact operator isospectral to the main one and zero-counting techniques.  Theorem\,\ref{Theo2} is proved in \S\ref{pf2}. 

\subsection{Proof of Theorem\,\ref{T1}(i)}\label{pfT1}

Denoting 
\begin{equation}\label{Phig}
p=-iq,\ \ \ \F(p)=g(q)
\end{equation}
 \eqref{eq7} becomes
\begin{equation}
  \label{eq:rec11}
  g(q)=\a h(q+\omega)g(q+\omega)-\a h(q-\omega)g(q-\omega)+f(q)
\end{equation}
where
$$ h(q):=\frac{1}{2}\frac{1}{\sqrt{q-1}-i},\ \ \ 
f(q):=-\frac{\a\omega}{q^2-\omega^2}$$
It turns out that the pole of $h$ at $q=0$ has no bearing on the regularity of the  solutions, as the equation can be regularized in a number of ways.  One is presented in detail in \cite{CMP}. A simpler way is presented in \S\ref{simpler}.

It is convenient to discretize \eqref{eq:rec11}. With the notation
\begin{equation}
  \label{notaq}
  q=:q_n=\sigma+n\omega\ \ \text{with }{ \Re}\sigma\in[0,\omega),\ n\in\ZZ
\end{equation}
 and setting $ h_n=h(q_n)$, $f_n=f(q_n)$, we obtain the difference equations with parameters $\sigma,\omega$
\begin{equation}\label{Kdesigma}
g_n=\a h_{n+1}g_{n+1}-\a h_{n-1}g_{n-1}+ f_n 
\end{equation}
or, in operator notation,
\begin{equation}\label{Kdesigmaop}
\gg=K_0(\sigma,\alpha)\gg+\ff 
\end{equation}

\subsubsection{Regularization of the operator} \label{simpler}

We rewrite \eqref{Kdesigma}. Let
\begin{equation}
  \label{eq:not}
 d_n=\tfrac{1}{4}(\sigma+n\omega-i\beta)^{-1}\text{ and  }b_n= d_n/(h_{n+1} h_{n-1});\ \ \ \ (\beta>0)
\end{equation}
Then,
\begin{equation} \label{eq:rec201}
  g_n= \left(1- b_n\right)g_n+\frac{\a d_n}{h_{n-1}} g_{n+1}-\frac{\a d_n}{h_{n+1}}g_{n-1}+ b_n f_n 
  \end{equation}
or,
\begin{equation} \label{eq:rec201op}
  \gg=K(\sigma,\a) \gg+B\ff
\end{equation}
where $(Bf)_n=b_nf_n$. Now $K$, $\ff$ and $B$  are pole-free in the closed lower half plane (analyticity of the solution  in the upper half plane is known, see the beginning of \S\ref{prior}). Note that $\ff$ is a multiple of $\a$.

\z{\bf Extension.} It is convenient to remove the restriction in \eqref{notaq} on $\sigma$, and allow $\sigma\in\CC$.

\begin{Remark}\label{KK0}
Note that $K=I-B+BK_0$ therefore $(I-K_0)^{-1}=(I-K)^{-1}B$, so equations \eqref{Kdesigmaop} and \eqref{eq:rec201op} are equivalent wherever $B$ is invertible, that is, for $\sigma\not\in 1+\omega\ZZ$. We will rely on this equivalence to choose the more convenient one for a particular purpose.

\end{Remark}

\begin{Proposition}
  (i) The operator $K(\sigma,\a)$ is  compact in $\ell^2(\ZZ)$. It is linear-affine in $\a$, and analytic in $\sigma$ except for a square root branch point at  $1-\lfloor \omega^{-1}\rfloor\omega$.
  
  (ii) For   $\sigma\ne 0$, $K_0(\sigma,\alpha)$ has the properties of $K$ listed above.
\end{Proposition}
\begin{proof}
  (i) For compactness, note that $K(\sigma,\a)$ is a composition of two shifts and  multiplications by diagonal operators whose elements vanish in the limit $|n|\to\infty$ (all its coefficients, see \eqref{eq:rec201},  are $O(|n|^{-1/2})$). Noting that $h(q)$ has a pole at $q=0$ and a branch point at $q=1$, the analyticity properties are manifest.
 The proof of (ii) is similar.
\end{proof}

Theorem\,\ref{T1}(i) now follows from the results above, Proposition\,\ref{Prop2} below,  and an argument similar to (and simpler than) the one in \S\ref{Theo2}.

\begin{Proposition}\label{Prop2}
The homogeneous equation
\begin{equation}
  \label{eq:rec1h}
   \gg=K(\sigma,\a)\gg
\end{equation}
has no nontrivial $\ell^2$ solution if $\Im\sigma\Ge 0$. By the Fredholm alternative \eqref{eq:rec201op} has a unique solution which has the same analyticity properties as $K$.  

In particular, $(I-K(\sigma,\a))^{-1}$ is analytic if $\Im\sigma>0$ and on each segment $(1+n\omega,1+(n+1)\omega),\ n\in\ZZ$; at $\sigma=\beta_n=1+n\omega$ it is analytic in $\sqrt{\sigma-\beta_n}$.

\end{Proposition}

The proof is given in  \cite{OCJLAR}; for completeness, we sketch the argument in the Appendix.

\subsection{Further properties of the homogeneous equation}\label{PropHomEq}
The general theory of recurrence relations \cite{Immink} shows that the homogeneous part of \eqref{Kdesigma}  has two linearly independent solutions, one that grows like $(\a/2)^{-n} (n!)^{1/2}$ and one that decays like $(\a/2)^n/ (n!)^{1/2}$ for $n\to \infty$, and two similar solutions for $n\to-\infty$; the one that decays at $+\infty$ is different from  the one that decays at $-\infty$, unless $1$ is in the  $\ell^2$ spectrum of $K_0(\sigma,\a)$. Since we need more details, we reprove the relevant claims.  The main results are given in Corollaries\,\ref{Cconv},\,\ref{Cgenrec}. 

In this section it is convenient to work with the continuous equations  \eqref{eq:rec11}.
Its homogeneous part is
\begin{equation}
  \label{eq:rec2n}
  g(q)=\a h(q+\omega)g(q+\omega)-\a h(q-\omega)g(q-\omega)
\end{equation}

Lemma\,\ref{Lconfrac} shows the existence of a solution of \eqref{eq:rec2n} which goes to zero as  $n\to +\infty$ for $\a$ not too large, with tight uniform estimates for all $q\in\RR$, and of a similar solution for $n\to -\infty$. Lemma\,\ref{Llargeq} shows existence of such solutions for any $\a>0$, providing estimates only for $|q|$ large enough.

Looking for a solution that decays for large $q$ we define 
$$\rho(q):=\rho(q;\a)=g(q)/g(q-\omega)  $$
 and obtain from \eqref{eq:rec2n}
\begin{equation}
  \label{eq:eqconfrac}
  \rho (q)=\mathcal{N}(\rho)
\end{equation}
where $\mathcal{N}$ is the nonlinear operator
$$\mathcal{N}(\rho)(q)=-\frac{\a h(q-\omega)}{1-\a h(q+\omega) \rho (q+\omega)}$$

Similarly, looking for solutions which decay for $q\to -\infty$ the ratio
$$\Omega(q):=\Omega(q;\a)=g(q-\omega)/g(q)$$
 satisfies
\begin{equation}
  \label{eq:eqconfracNeg}
  \Omega \left( q \right) =\mathcal{M}(\Omega)
\end{equation}
where
$$\mathcal{M}(\Omega)(q)={\frac {\a h \left( q \right) }{1+\a h \left( q-2\,
\omega \right) \Omega \left( q-\omega \right) }}  $$

\begin{Notations}\label{Nota}
As usual, a domain in $\CC$ is an open, connected subset.
$\HHL$ denotes the open lower half plane in $\CC$.

 Let 
$$A>0,\ \ \ \ \a_A=A/(1+A^2)$$
 and denote   
$$J_{N}=\left\{ q\, |\, \Re(q)\Ge N\omega, \, \Im (q)\in[0,\epsilon]\right\}$$
 (for a suitably small $\epsilon$). Consider the Banach space $\mathcal{C}(J_{N})$ of functions continuous in the strip $J_{N}$, with the sup norm. Let $\mathcal{C}(\tilde{J}_{-1})$ 
denote the Banach space of continuous functions in 
$$\tilde{J}_{-1}:=\{ q| \Re(q)\Le -\omega, \Im (q)\in[0,\epsilon]\}$$
We denote by  $\mathcal{R}$ the class of functions which are real-analytic in $\a$ for all $\a\in\RR^+$ and in $\q\in\HHL$, continuous on $\overline{\HHL} $ and with possible square root branch points at $q\in 1+\omega\ZZ$. 

\end{Notations}

\begin{Remark}\label{cuts} {\textsl By the usual properties of the Laplace transform, $g$ is analytic in the upper half plane. Since below we are interested in the properties of $g$ in the lower half plane it is convenient to place the branch cuts in the upper half plane. 
Later, in \S\ref{bccon}, when we deform the contour of an inverse Laplace transform (in $q$ it is horizontal, in the upper half plane), the points on the curve are moved vertically down, yielding a collection of vertical Hankel contours \footnote{A Hankel contour is a path surrounding a singular point, originating and ending at infinity, \cite{Krantz}.} around the branch points $1+n\omega$ and residues. For this particular purpose, placing the  cuts in the upper or lower half plane can be seen to be  equivalent.}  
\end{Remark}

\begin{Lemma}\label{Lconfrac}

(i)  For $|\a|<\a_A$, the operator $\mathcal{N}$ defined in \eqref{eq:eqconfrac} is contractive in the ball $\|\rho\|<A$ in $\mathcal{C}(J_{N})$; the contractivity factor is $\tfrac{1}{4}\a^2(1+o(1))$ as $\a\to 0$.  

Thus  \eqref{eq:eqconfrac} has a unique fixed point $\rho\in \mathcal{C}(J_{N})$. Also,   $\rho$ is analytic in $\a$ for $|\a|<\a_A$ and satisfies $|\rho(q)|\Le \frac12 \a (1+o(1))$ as $\a\to 0$.  


(ii) The operator $\mathcal{M}$ is contractive in a ball $\|\Omega\|<A$ in $\mathcal{C}(\tilde{J}_{-1})$.
\end{Lemma}

 Let us state first a more general result, valid for all $\a\in\RR^+$ (where now  the dependence of $\rho$ on $\a$ is made explicit): 

 \begin{Lemma}\label{Llargeq}
 (i)   For any fixed $\a_0$ and large enough $q_0>0$ the operator $\mathcal{N}$ in \eqref{eq:eqconfrac} is contractive in the ball
    $$B=\{\rho\,|\, \sup_{ q\Ge q_0,\ |\a|\Le\a_0}|\rho(q,\a)| \Le \a_0q_0^{-1/2} \}$$
   and thus it has a unique solution, which is analytic in  $(q,\a)$.
   
  Similar estimates hold for \eqref{eq:eqconfracNeg}.
    
 (ii)  For large $q>0$,
    \begin{equation}
      \label{eq:531}
      \rho(|q|;\a)= -\tfrac{1}{2}\a |q|^{-1/2}(1+o(1)), \ \   \Omega
     (-q)=-\tfrac{1}{2}\a q^{-1/2}(1+o(1))\ \ \ (q\to +\infty)
   \end{equation}
 
 (iii) For large $q$ in the lower half-plane,
   \begin{equation}
      \label{eq:532}
      |\rho(q;\a)|= \tfrac{1}{2}\a |\Im q|^{-1/2}(1+o(1)), \ \   |\Omega
     (q)|=\tfrac{1}{2}\a |\Im q|^{-1/2}(1+o(1))\ \ (\Im q\to -\infty)
   \end{equation}
         \end{Lemma}

\begin{proof}[Proof of Lemma\,\ref{Lconfrac}.] 

 (i) We note that the minimum of $|\sqrt{q-1}-i|$ in $J_{N}$ is $\sqrt{N\omega}+O(\epsilon)$ implying $\|h\|< \tfrac 12$ for $\epsilon$ small enough.

  We first show that $\mathcal{N}$ leaves the ball $\|\rho\|\Le A$ invariant. A straightforward estimate shows that if $|\a|\Le \a_0$ then $\mathcal{N}$ is well defined on the ball and
$$\|\mathcal{N}(\rho)\|<  A$$
The contractivity factor is obtained by taking the sup of the norm of the Fr\'echet derivative of $\mathcal{N}$ with respect to $\rho$:
\begin{equation}
  \label{eq:difn}
  \left\|\frac{\partial \mathcal{N}}{\partial\rho}\right\|=\left\|\frac{\a^2T^{-1} h Th}{(1-\a T(h\rho))^2}\right\|
\end{equation}
where
\begin{equation}
  \label{eq:deft}
  \ \ \ \ (Tg)(q):=g(q+\omega)
\end{equation}
Under the assumptions in the Lemma, we have
\begin{equation}
  \label{eq:difn2}
  \left\|\frac{\partial \mathcal{N}}{\partial\rho}\right\|\Le   \frac{\a^2\|h\|^2}{(1-\a\|h\|A)^2} <1,\ \text{and}\ \left\|\frac{\partial \mathcal{N}}{\partial\rho}\right\|< \tfrac14 \a^2 (1+o(1))\ \text{as}\ \a\to 0
\end{equation}

The same analysis goes through in the space of functions $\rho(q,\a)$ which are of class $\mathcal{R}$ for $q\in J_N$ and analytic in $\a$ for $|\a|\Le \a_{A}$, in the joint sup norm, in $\a$ and $q$, proving joint analyticity in $\a,q$, except for the mentioned square root branch points. To show that $q=1+n\omega$ are square root branch points, we return to the $\sigma+n\omega$ representation of \eqref{notaq}. For simplicity of presentation assume $\omega>1$. We repeat the arguments above, now in the space of functions of the form $A_1\sqrt{\q-1}+A_2$ where $A_1,A_2$ are analytic near $\q=1$, in the norm $\|A_1\|+\|A_2\|$.

Moreover, since the only singularities of $h(q)$ are a pole of order one at $q=0$ and a square root branch point at $q=1$, a similar analysis shows that $\rho$ is of class $\mathcal{R}$.

Straightforward estimates in \eqref{eq:eqconfrac}  show that for small $\a$ we have
$$\|\rho\|<\tfrac{1}{2}\a(1+o(1))$$

(ii) We note that  $\max_{\tilde{J}_{-1}}|h(q)|=|h(-\omega)|$ hence $\|h\|<\tfrac 12$. The rest of the proof is as for (i).
\end{proof}

\

\begin{proof}[Proof of Lemma\,\ref{Llargeq}.]

 We note that for any $A$ and $\a$, we have, for large enough $q$ 
\begin{equation}
  \label{eq:25}
\| \mathcal{N}(\rho)\| \Le \frac 12 \ \a q^{-1/2}(1+o(1))\ \text{and}\ \left\|\frac{\partial \mathcal{N}}{\partial\rho}\right\|\Le \frac 14\ \a^2 q^{-1}(1+o(1))
\end{equation}

\end{proof}

We use Pringsheim's notation for continued fractions $A+B/(C+D/(E+\cdots))=A + \cFrac{B}{C} + \cFrac{D}{E}+\cdots$
\begin{Corollary}\label{Cconv}
 Iteration of \eqref{eq:eqconfrac} yields a continued fraction
\begin{equation}
    \label{eq:24} \rho(q)=- \cFrac{\a h(q-\omega)}{1}+\cFrac{\a^2\ h(q+\omega)h(q)}{1}+\cFrac{\a^2\ h(q+2\omega)h(q+\omega)}{1}+\cdots \end{equation}
which is convergent for $\Re(q)\Ge N\omega$, $\Im q \in[0,\epsilon]$ and  $\a<2/3$. For small $\a$, the rate of convergence is $4^{-n}\a^{2n}$.

Similarly, iteration of \eqref{eq:eqconfracNeg} yields a convergent continued fraction
  \begin{equation}
   \label{eq:24O}
\Omega (q)=  \cFrac{\a h(q)}{1}+    \cFrac{\a^2 h(q-\omega)h(q-2\omega)}{1}+\cFrac{\a^2\ h(q-2\omega)h(q-3\omega)}{1}+\cdots
 \end{equation}
\end{Corollary}
 \begin{proof}
   Convergence of the continued fraction, by definition, means that the $n-$th truncate of the continued fraction, that is  $\mathcal{N}^{\circ n}(0)$,  converges to the fixed point $\rho=\mathcal{N}(\rho)$. Since zero is in the domain of contractivity of $\mathcal{N}$, convergence follows directly from Lemma \ref{Lconfrac}. The norm of the Fr\'echet derivative of   $\mathcal{N}^{\circ n}$ is $(\tfrac14 \a^2)^n(1+o(1))$ implying the last statement. Convergence of \eqref{eq:24O} is similar.
    \end{proof}
    
    \begin{Corollary}\label{Cgenrec}
   As $n\to -\infty$, there is a solution of \eqref{Kdesigma} with $f=0$ which is $O(2^{-|n|}\a^{|n|} |n|^{-|n|/2})$; a second,  linearly independent solution, has the property $1/g_n=O(2^{-|n|}\a^{|n|} |n|^{-|n|/2}))$, that is, such a solution grows factorially. A similar statement holds as $n\to \infty$. 
 \end{Corollary}
 \begin{proof}
   The first part follows from the fact that $g_{n-1}/g_n=\Omega_n$. For the second part, one looks as usual for a second solution in the form $h_n=g_nu_n$ and notes that $u_n$ satisfies a first order recurrence relation that can be solved in closed form in terms of $g_n$.

 \end{proof}

  \subsection{Proof of Theorem\,\ref{T4}\,(i): location of the singularities for small $\a$}\label{polsmalr} 
  
  By Proposition\,\ref{Prop2} the resolvent can only be singular if $\Im \sigma<0$, which we will assume henceforth. By Remark\,\ref{KK0} we can then work with the simpler operator $K_0$. We place branch cuts in the upper half plane, see Remark\,\ref{cuts}. 
 
  \begin{Theorem}\label{L3} There is a $\delta>0$ such that for all complex $\a$ with $|\a|<\delta$ the following hold.
  
 (i) There exists a unique $\sigma=\q_0(\a)$ in the strip $\Re\sigma\in[0,\omega)$ so that $\mathop{Ker}(I-K(\sigma,\a))\ne\{0\}$. 

 More precisely, $q_0(\a)=\a^2 s_0(1+\a f(\a))$ where
\begin{equation}\label{vals0}
 s_0=-\frac {\sqrt {1+\omega}+i\sqrt {\omega-1}}{2\omega}
\end{equation}
for some $f$ analytic at zero.

(ii) For $\omega\in(\frac 1m,\frac 1{m-1}),\, m-1\in\NN$ we have\footnote{ Note that for $\omega<1$ we have $s_0=-(2\omega)^{-1}(\sqrt {1+\omega}-\sqrt {1-\omega})\in\RR$.}
\begin{equation}\label{valsomincm}
q_0(\a)=\a^2 s_0\, (1+o_a)+i\a^{2m}\xi_0\, (1+o_a)
\end{equation}
where $s_0$ is real, given by \eqref{vals0}, and $\xi_0$ is real, given by 
\begin{equation}\label{valxi0}
 \xi_0= -\frac{1}{m\omega}\frac{\sqrt{m\omega-1}}{\prod_{1\Le k\Le m-1}(1-\sqrt{1-k\omega})^2}\frac{1}{2^{2m+1}} 
\end{equation}

\end{Theorem}

{\bf Remark} Since $p=-iq$,  the poles in the $p$-plane are at $-iq_0(\a)+i\omega\ZZ$, hence Theorem\,\ref{L3} completes the proof  of Theorem\,\ref{T4}\,(i).

For the proof of Theorem\,\ref{L3} we first show, in Lemma\,\ref{Lsmallr}, that any singularities of $(I-K(\sigma,\a))^{-1}$ are $O(\alpha^2)$ distance to $ \omega\ZZ$, if $\a$ is small enough. Then location of the poles is found by series expansions in $\a$.

\begin{Lemma}\label{Lsmallr}
There is a $C(\omega)>0$ such that for $|\a|$ small enough equation \eqref{eq:rec201} has a unique
solution for any $ \sigma$ in the strip $\Re\sigma\in[0,\omega)$ with 
dist$(\sigma,\{0,\omega\})\Ge C(\omega)\a^{2}$. In particular, for such $ \sigma$, {\em Ker}$(I-K(\sigma,\a))=\{0\}$.

\end{Lemma}

\begin{proof} As mentioned above, we can additionally assume that $\Im \sigma<0$, 
hence invertibility of $I-K( \sigma,\a)$ is equivalent to
  Ker$(I-K( \sigma,\a))=\{0\}$, which is equivalent to  Ker$(I-K_0( \sigma,\a))=\{0\}$.

Let $K_0( \sigma,\a)=\a K_1 ( \sigma)$ (where now $K_1$ does not depend on $\a$). If $u$ is such
  that $(I-K_0( \sigma,\a))u=0$, then $u=K_0u=\a K_1u= \a^2K_1^2u$. But straightforward
  estimates show that $\|\a^2K_1^2\|<1/2$, if $C$ is large enough and 
  dist$( \sigma,\{0,\omega\})>C\a^{2}$. This implies that $\a^2K_1^2$
  is contractive and thus $u=0$.
\end{proof}

  \begin{Proposition}\label{rhomero}
        The functions $\rho(q,\a),\, \Omega(q,\a)$ are meromorphic in $(q,\a)$ for $q$ in the open lower half plane and $\a\in\CC$.
        \end{Proposition}
        
        \begin{proof} 
        
          Let $\a$ be fixed and $\rho=\rho(q;\a)$ be the fixed point provided by Lemma \ref{Llargeq}, analytic in $(q,\a)$ for $q>q_0$. Using the recurrence relation \eqref{eq:eqconfrac} $\rho$ can be continued to a meromorphic function for all $q$ with smaller real part (the coefficients in \eqref{eq:eqconfrac} are meromorphic except for square root branch points on the real line). 
          
          Similarly, $\Omega$ can be continued to a meromorphic function.
        \end{proof}

\begin{proof}[Proof of Theorem\,\ref{L3} (i)]
 Let $\rho,\ \Omega$ be given by Proposition\,\ref{rhomero}. The value(s) of $\sigma=\sigma(\a)$ for which Ker$[I-K(\sigma,\a)]\ne\{0\}$ are those for which 
\begin{equation}
  \label{eq:root}
  \rho(\sigma,\a)=\frac{1}{\Omega(\sigma,\a)}
\end{equation}
as discussed in \S\ref{PropHomEq}.
By Lemma\,\ref{Lsmallr} any such $\sigma$ has the form $\sigma=\a^2 s$ with $|s|< C(\omega)$. We now show that for small $\a$ there is exactly one solution of \eqref{eq:root} in the strip $\Re \sigma\in [0,1-\lfloor\omega^{-1}\rfloor\omega)$. (Note that $0<1-\lfloor\omega^{-1}\rfloor\omega<\omega$ using the assumption that $\omega^{-1}\not\in\ZZ$ which ensures that the poles and the branch points do not coincide.)

Expanding $\rho$, respectively $\Omega$, in a power series in $\a$ we have
\begin{equation}\label{match}
\rho(\a^2s,\a)=\frac{i\a }2\,\frac 1{\sqrt{\omega+1}-1}\,\frac 1{1-\frac 1{2s}\,\frac 1{i\sqrt{\omega-1}+1}} (1+O(\a^2)),\ \ \ \ \ \text{and}\ \ \ \ \ \Omega(\a^2s,\a)^{-1}=-i\a s (1+O(\a^2))
\end{equation}
Let $F(r,s):= \rho(\a^2s)-\Omega(\a^2s)^{-1}$. It can be checked that $F$ is analytic in $\{(s,\a)| |s|<C(\omega),|\a|<\delta\} $ for small enough $\delta$.
Equating the dominant terms in \eqref{match} we obtain that $F$ has only one simple zero, of the form \eqref{vals0}.

 Finally, note that if equation \eqref{eq:root} had a solution of the form $\sigma=\omega-\a^2 s_1$ then $\sigma=-\a^2 s_1$ would also be a solution, but this is ruled out by the uniqueness of $O(\a^2)$  solutions.
\end{proof}

\begin{proof}[Proof of Theorem\,\ref{L3} (ii)]

As in the proof of Corollary\,\ref{Cconv}, the $n-$th truncate of the continued fraction defining $\rho$, that is  
$\mathcal{N}^{\circ n}(0)$,  converges to the fixed point $\rho=\mathcal{N}(\rho)$. Similarly, $\mathcal{M}^{\circ n}(0)$, converges to the fixed point $\Omega=\mathcal{M}(\Omega)$. 

With the notation
$$x_k=\a^2\,t_k,\ \ \ \ \ t_k=h(\sigma+k\omega)\, h\left(\sigma+(k-1)\omega\right),\ \ \ \ \  k\in\ZZ$$
we obtain from \eqref{eq:24},\ \eqref{eq:24O}
 \begin{align}\label{apprrho1}
\mathcal{N}^{\circ (m+1)}(0)= -\ \cFrac{\a h(\sigma-\omega)}{1}+\cFrac{x_1}{1}+\cFrac{\a^2t_2}{ 1}\cdots+\cFrac{\a^2 t_m}{1},\nonumber\\
 \mathcal{M}^{\circ (m+1)}(0)=  \cFrac{\a h(\sigma)}{1}+\cFrac{\a^2t_{-1}}{1}+\cFrac{\a^2t_{-2}}{ 1}\cdots+\cFrac{\a^2t_{-m} }{1}
\end{align}

Note that, for small $\a$, $x_1=O(1)$ and all other $x_k=O(\a^2)$ and, inductively, $\mathcal{N}^{\circ (k+1)}(0)=\mathcal{N}^{\circ (k)}(0)+O(\a^{2k-1})$ and $\mathcal{M}^{\circ (k+1)}(0)=\mathcal{M}^{\circ (k)}(0)+O(\a^{2k+1})$. Therefore
$$\rho(\sigma)=\mathcal{N}^{\circ (m+1)}(0)+O(\a^{2m+1}),\ \ \Omega(\sigma)=\mathcal{M}^{\circ (m+1)}(0)+O(\a^{2m+3})$$

Noting that
$$h(\sigma+k\omega)=\frac 1{2i}\, \frac 1{\sqrt{1-\sigma-k\omega}-1}\ \ \text{for }k<m$$
it follows that all $t_k$ with $k<m$ have a power series expansion in $\a$ with real coefficients
for $\omega\in (\frac 1m,\frac 1{m-1})$. Then so does the denominator in the right of \eqref{apprrho1}, as well as the denominator on the left,  truncated to $m-1$ terms. 

We now look for a solution to
\begin{equation}\label{eqappr}
1=\rho(\sigma)\Omega(\sigma)=\mathcal{N}^{\circ (m+1)}(0)\mathcal{M}^{\circ (m+1)}(0)(1+O(\a^{2m+1})) 
\end{equation}
in the form $\sigma=\a^2s+O(\a^{2m+2})$ where $s=s_0+\a^2s_1+\ldots+i\a ^{2m-2}\xi$ with $\xi$ and all $s_k$ real.

A simple calculation gives
 \begin{equation}\label{valqm}
 t_m=\mathcal{P}+iq_m+O(\a^2),\ \ \text{where } q_m=-\frac{\sqrt{m\omega-1}}{4m\omega}\,\frac 1{\sqrt{1-(m-1)\omega}-1} ,\  \ \mathcal{P}\in\RR
 \end{equation}
which implies
$$\frac{\a^2t_{m-1}}{1+\a^2t_m}=\mathcal{P}-i\a^4t_{m-1}q_{m}+O(\a^6)$$
where here, and in the following, $\mathcal{P}$ denotes quantities that depend polynomially on $\a$ and have real coefficients. Then, inductively, we obtain that 
\begin{multline}\label{valQ}
\cFrac{1}{1}+\cFrac{\a^2t_2}{1}+\cdots \cFrac{\a^2t_m}{1}=1+\a^2 \mathcal{P}+i (-1)^{m-1}\a^{2m-2}t_2...t_{m-1}q_{m}+O(\a^{2m})\\
:=1+\a^2\mathcal{P}+i\a^{2m-2}Q+O(\a^{2m})
\end{multline}
Equation \eqref{eqappr} becomes
\begin{equation}\label{eqappr2}
\cfrac{-\a h(\a^2s-\omega)}{1+x_1(1+\a^2\mathcal{P}+i\a^{2m-2}Q+O(\a^{2m}))} \, \frac{\a h(\a^2s)}{1+\a^2\mathcal{P}}=1+O(\a^{2m+1})
\end{equation}
We have 
\begin{equation}\label{forma1}
-\a^2  h(\a^2s)h(\a^2s-\omega)= -\frac{1}{2s}\frac 1{\sqrt{1+\omega}-1}    +O(\a^2)
\end{equation}
\begin{equation}\label{forma2}
x_1= \a^2  h(\a^2s)h(a^2s+\omega)=  \frac{1}{2s}\frac 1{\sqrt{1-\omega}-1}  +O(\a^2)
\end{equation}
Using \eqref{forma1},  \eqref{forma2} in \eqref{eqappr2} and equating the coefficient of $\a^0$ we obtain \eqref{vals0}.

Rewriting  \eqref{eqappr2} using \eqref{forma1},  \eqref{forma1} as
$$\frac{ \frac 1{\sqrt{1+\omega}-1}    +sO(\a^2)}{ 2s+(\frac 1{\sqrt{1-\omega}-1}  +sO(\a^2))  (1+\a^2\mathcal{P}+i\a^{2m-2}Q+O(\a^{2m}))   }\frac{1}{1+\a^2\mathcal{P}}  =1+O(\a^{2m+1}) $$
and equating the dominant imaginary terms, of order $O(\a^{2m-2})$, we obtain $2\xi+ \frac Q{\sqrt{1+\omega}-1}=O(\a^{2m})$.  
Using \eqref{valQ},\, \eqref{valqm} and noting that
 $$t_k= -\frac 14\frac 1{\sqrt{1-k\omega}-1}\,\frac 1{\sqrt{1-(k-1)\omega}-1}+O(\a^2)$$  
 we obtain formula \eqref{valxi0}.

    \end{proof}

  \subsection{Proof of Theorem\,\ref{T1}(ii), (iii): structure of the resolvent $(I-K(\sigma,\a))^{-1}$ for any $\a\ne 0$ }\label{Sec5}

    \begin{Theorem}\label{T-npoles}
 
   Let $\a\in\RR,\ \a\ne 0$. $(I-K(\sigma,\a))^{-1}$ is analytic for $\sigma\in\HHL$ except for one array of simple poles and it is continuous on $\overline{\HHL}$. 
      
   The poles are located at $\sigma=\q_0(\a)+n\omega$ (for all $n\in\ZZ$) where $\Im \q_0(\a)<0$ and $\q_0(\a)$ is real-analytic for $\a>0$.
   
   As a consequence $(I-K(\sigma,\a))^{-1}$ is analytic in a strip $\Re\sigma\in[0,\omega),\ \Im\sigma<0$ except for one simple pole.   
  
  \end{Theorem}
  
 {\bf Note} that going back to the variable $p$, Theorem\,\ref{T-npoles}  implies the first statement of Theorem\,\ref{T1}(ii) and \eqref{eq:42}.
 
 \
 
 The structure of the proof of Theorem\,\ref{T-npoles} is as follows. The results were proved for small $\a$ in \S\ref{polsmalr}, see also Lemma\,\ref{L30}\,(i). We extend them to all $\a\ne 0$ using a general result on the constancy of number of zeros of analytic, periodic functions depending on a parameter contained in Lemma\,\ref{Larg}. To apply this Lemma, we construct a periodic operator  isospectral to $I-K(\sigma,\a)$ (Lemma\,\ref{PropK2}). This is especially convenient since working in the whole lower half plane would mean working with infinitely many poles while restricting $\sigma$ to a strip introduces a number of unnecessary complications.
 
Note that continuity up to the real line follows from Proposition\,\ref{Prop2}. Therefore it suffices to consider $\Im\sigma<0$, and by Remark\,\ref{KK0}, we can work with the operator $K_0$.

   \begin{Lemma}\label{Larg}
   Let $A,Q,\omega>0$ and define
\begin{equation}
  \label{eq:S}
S=\{ z\, |\,\Re z \in (0,\omega),\Im z<0\};\ \ I_A=(0,A)
\end{equation}

 Let $H(z,\a)$ be a function which is real-analytic in  $\a\in I_A$, analytic and exponentially bounded in  $z \in S$. Assume further that $H$ is periodic,  $H(z+\omega,\a)=H(z,\a)$, continuous in $\overline{S}\times I_A$, and
    \begin{equation}
      \label{eq:cond}
  H(z,\a)\ne 0 \ \ \text{ if }\ \ \  (z,\a)\in \{z,\Im z\Le -Q\text{ or }\Im z=0\}\times I_A    
   \end{equation}
Let $Z(\a)$ be the number of zeros,  counting multiplicity, of $H$ in $\overline{S}\times I_A$. 

Then the function $Z$ is constant. If $Z=1$, then defining $q_0$ by $H(q_0(\a),\a)=0$, $q_0$ is real-analytic in $\a\in I_A$. 
 \end{Lemma}
\begin{proof}
  Multiplying $H$ by $e^{-2\pi i N z/\omega}$ for some  $N\in\NN$ we can arrange that $H\to 0$ as $\Im z \to -i\infty$. By \eqref{eq:cond} and continuity, there is an $\epsilon=\epsilon(A)$ so that for all $\a<A$ all the zeros of $H(\cdot,\a)$ are in  $\{\Im z \in (-Q,-\epsilon)\}$.
  
  Fix $\a\in (0,A)$ and choose a small  $\beta\ge 0$ so that $H\ne 0$ if  $q$ is on $\partial S_\beta$, where  $S_\beta=S+\beta-i\epsilon$.  By the argument principle, the number of zeros in $S$ counting multiplicity is
\begin{equation}\label{zeros}
Z(\a)=-\frac{1}{2\pi i}\int_{-i\epsilon/2}^{\omega-i\epsilon/2}\frac{\partial_s H(s,\a)}{H(s,\a)}\, ds
\end{equation}
noting that by periodicity the contributions of the vertical sides of  $\partial S_\beta$  cancel out and the integral over $[\beta-i\epsilon,\omega+\beta-i\epsilon]$ equals the integral in \eqref{zeros}.

The right side of \eqref{zeros} is manifestly real-analytic in $\a$ and integer-valued, thus constant.  Real-analyticity when $Z=1$ is an immediate consequence of the implicit function theorem. 
\end{proof}

Let $T$ denote the forward shift in $\ell^2$  (cf. \eqref{eq:deft}). A straightforward calculation shows that for any $\sigma$
\begin{equation}\label{L-Floquet}
K_0(\sigma+\omega,\a)=T K_0(\sigma,\a) T^{-1}
 \end{equation}

  \begin{Lemma}\label{PropK2}
  The operator
  $$K_2(\sigma,\a):=T^{-\sigma/\omega}K_0(\sigma,\a)T^{\sigma/\omega} $$
   is periodic in $\sigma$. 
   
   $I-K_2(\sigma,\a)$ is a periodic operator isospectral to $I-K_0(\sigma,\a)$ and $\|K_2\|=O(e^{2\pi|\Im \sigma|/\omega})$ for large $\sigma$.
  \end{Lemma}
  \begin{proof}
   
 Note that $T$ is a unitary operator, and thus $T=e^{iA}$ for some self-adjoint bounded operator $A$.  In $L^2(\mathbb{T})$, the Fourier transform space,  $T$ is multiplication by $e^{-i\phi}$ and $A$ is multiplication by $-\phi$. This reduces the analysis to a compact analytic manifold. 
   
    \end{proof}
    \begin{Lemma}\label{L21} 
For small $|\a|$ there is a unique pole of $(I-K_0(\sigma,\a))^{-1}$ in the strip $\Re q\in [0,\omega)$. The pole is simple and  analytic in $\a$.
  \end{Lemma}
\begin{proof}     
  Since the position of the unique pole is analytic for small $\a$, and is manifestly simple (see \eqref{eq:rec201}) when $\a=0$, this follows from the argument principle.
 \end{proof}  
   \begin{Lemma}\label{L30}
  (i)   For small $\a>0$ the resolvent $(I-K_0(\cdot,\a))^{-1}$ has only one array of poles,  located in the lower half plane at $\sigma=\a^2 s_0(1+\a g(\a))+n\omega,n\in\ZZ$ with $g$ analytic.

  (ii) For any $\a,q$,  $\text{dim Ker}(I-K_0(\sigma,\a))\in \{0,1\}$ and $\text{dim Ker}(I-K_0(\sigma,\a)^*)\in \{0,1\}$, where $*$ denotes the adjoint.
\end{Lemma}
\begin{proof}
  (i) is an immediate corollary of Lemma \ref{L21} and Theorem \ref{L3}.
  
  (ii) follows from Corollary \ref{Cgenrec}. Indeed, if $\gg\in \mathop{Ker} (I-K_0)$ then $\gg$ is an $\ell^2$ solution of the linear recurrence \eqref{eq:rec1h} which cannot have a two dimensional space of solutions decaying at both  $n\to\pm\infty$ by Corollary\,\ref{Cgenrec}. 
 
 Similar arguments apply to $\mathop{Ker} (I-K_0^*)$, since $K_0^*\yy=\a\overline{h}(T^{-1}-T)\yy$ and $K_0^*\yy=\yy$ yields a second order difference equation similar to that of $K_0$, having two solutions $O((n!)^{1/2}(-2\sqrt{\omega}/\a)^n)$, respectively $O((n!)^{-1/2}(\a/2\sqrt{\omega})^n)$ for $n\to\infty$ and two similar solutions (one decreasing to $0$ and another one increasing) as $n\to-\infty$.
\end{proof}

\begin{Proposition}\label{P-large-q}
For any $A>0$ there is $Q,\epsilon>0$ such that the Neumann series
  \begin{equation}
    \label{eq:Neumann}
    (I-K(\sigma,\a))^{-1}=\sum_{m=0}^\infty K(\sigma,\a)^m 
  \end{equation}
  converges to an operator valued function, analytic in $\{(\sigma,\a)|\Re\a\in [0,A), |\Im\a|<\epsilon, \Im \sigma<-Q\}$. 
  
  A similar statement holds for $K_0(\sigma,\a)$.
\end{Proposition}
\begin{proof}
  This simply follows from the fact that the shift operators have norm one, $h(q)=O((\Im q)^{-1/2})$, and are analytic in this regime.
\end{proof}

 \begin{proof}[Proof of Theorem \ref{T-npoles}] By Lemma\,\ref{PropK2} it suffices to prove these results for $(I-K_2(\sigma,\a))^{-1}$.

Let $F_k$ be finite rank operators converging to $K_2$ as $k\to\infty$ and $P_k$ the projectors on the range of $F_k$.
 Let $A>0$ and  choose, cf. Proposition \ref{P-large-q}, $Q=Q(A)$ so that    $I-K_2$ is invertible if  $\Im \sigma<-Q$, and  $\a\in(0,A]$.
 Let $\epsilon>0$ be small enough and $k=k_A$ large enough so that for $\a\in [0,A+\epsilon]$ and  $q$ in $B$ where 
 $$B:=\{z|\Re z \in[-\epsilon,\omega+\epsilon], \Im z\in [-Q-\epsilon,0]\}$$  we have $\|K_2-F_k\|<\epsilon$. The easily checked identity (cf. in  \cite{Reed-Simon1} p. 202)
  \begin{equation}
    \label{eq:reed-simon}
    (I-K_2)^{-1}=(I-F)^{-1}(I-(K_2- F_k ))\ \ \text{where }F:=F_k(I-(K_2-F_k))^{-1}
  \end{equation}
implies $I-K_2$ is invertible iff $I-F$ is invertible. Now $F$ is finite rank and  by the usual Fredholm alternative  $I-F$   is not invertible iff
  $x=Fx$ has a nonzero solution. 
   Since $F=P_k F$, if $x=Fx$ we have $x=P_kx$. Thus the condition for $(I-K_2)^{-1}$ to have a pole in $S$ is  $h_A:=\text{det}\, M_A=0$ where $M_A$ is the matrix of $P_k(I-F)P_k$. 

  To end the proof, we note that $A$ is arbitrary, $h_A$ satisfies the conditions of Lemma \ref{Larg}, by Proposition \ref{Prop2} and analyticity of the matrix elements of $M_A$. This also completes the proof of the theorem. 
  \end{proof}  
  
  \subsection{Proof of Theorem\,\ref{T1}(iii)}\label{Sec555}
     The first part is an immediate consequence of Lemma \ref{L30}. The second statement follows from \eqref{eq:root} and the Proposition \ref{rhomero}.

\subsection{Proof of Theorem\,\ref{T4}\,(ii): calculation of the residues}

\subsubsection{General expression for residues and proof of \eqref{eq:43}}\label{abstract} The argument is based on a Laurent expansion of the resolvent and general properties of compact operators.

In this section we consider only $\sigma$ with $\Im\sigma<0$; therefore we can work with the operator $K_0$, see Remark\,\ref{KK0}.

Note that $Ran( I-K_0(\sigma,\alpha))$ and $Ran( I-K_0(\sigma,\alpha)^*)$ are closed, since $K_0(\sigma,\alpha)$, and therefore $K_0(\sigma,\alpha)^*$ are compact \cite{Linop}.

Let $q_0=q_0(\a)$ be as in Theorem\,\ref{T-npoles}. Then $\mathop{Ker} (I-K_0(q_0,\a))\ne\{0\}$, therefore it is one dimensional by Lemma\,\ref{L30}. Let $\yy_0$ be a unit vector generating this kernel.

Since $0$ belongs to the spectrum of $L_0:=I-K_0(q_0,\a)$ then it also belongs to the spectrum of $L_0^*=I-K_0(q_0,\a)^*$ 
This means that $1$ is an eigenvalue of the compact operator $K_0(q_0,\a)^*$, hence it is in the point spectrum. Let $\yy_0^*$ be a unit vector generating $\mathop{Ker} L_0^*$.

By Theorem\,\ref{T-npoles} equation \eqref{eq:rec201} has a solution $\gg$ which has a pole of order one when $z:=q-q_0=0$. Thus  $\gg=\gg_{-1}z^{-1}+\gg_0+z\tilde{\gg}$ where $\tilde{\gg}$ is analytic at $z=0$.

Since $K_0(\sigma,\a)$ is analytic in $\sigma$ at $\sigma=q_0$ we can write
$I-K_0(\sigma,\a)=L_0+zS(z,\a)$ with $S$ analytic. With this notation equation \eqref{Kdesigma} becomes
\begin{equation}\label{equa}
 L_0\left( z^{-1}\gg_{-1}+\gg_0+z\tilde{\gg}\right) +z S(z,\a)\left(z^{-1}\gg_{-1} +\gg_0+z\tilde{\gg}\right)=\ff
 \end{equation}
Since $\ff$ is analytic, this implies that $\gg_{-1}=\lambda \yy_0$ for some scalar $\lambda$.

Decompose $\ff=\la \ff,\yy_0^*\ra \yy_0^*+\ww$ where $\ww\in (\mathop{Ker}\,L_0^*)^\perp=Ran\,L_0$. With the notations  $S_0:=S(0,\a)=\frac{\partial L_0}{\partial q}|_{z=0}$ and  $\ff_0=\ff(q_0,\a)+O(z)$, equation \eqref{equa}  becomes
\begin{equation}\label{equa2}
 L_0\gg_0 +\lambda S_0\yy_0+O(z)=\la \ff,\yy_0^*\ra \yy_0^*+\ww
  \end{equation}
Noting that $L_0$ is invertible from $(\mathop{Ker}\,L_0)^\perp$ to $Ran\, L_0$ equation \eqref{equa2} is solvable only if 
$$\la \ff_0,\yy_0^*\ra \yy_0^*-\lambda S_0\yy_0\in Ran\, L_0=(Sp\, {\yy_0^*})^{\perp}$$
The argument above works for arbitrary $\ff_0\in\ell^2$, $\la S_0\yy_0,\yy_0^*\ra\ne 0$ and  $\lambda=\la \ff_0,\yy_0^*\ra/\la S_0\yy_0,\yy_0^*\ra$, thus
\begin{equation}
  \label{residue}
 \gg=\frac{1}{\sigma-q_0}\, \frac{\la \ff_0,\yy_0^*\ra}{\la S_0\yy_0,\yy_0^*\ra}\, \yy_0+O(1)\ \ \ (\sigma\to q_0)
\end{equation}

       \begin{Proposition}\label{residues} The residues $R=(R_n)_{n\in\ZZ}$ defined in Theorem\,\ref{T1}\,(ii) are multiples of $\yy_0$, given by \eqref{residue}, therefore they satisfy $(I-K_0)R=0$ and consequently $R_n=O(2^{-n}\a^n n^{-n/2})$, implying \eqref{eq:43}.
        \end{Proposition}
   
        \begin{proof}
          To see this we consider disks $\mathbb{D}_n$ around the  poles of $\F$ small enough to contain no pole of the nonhomogeneous part of the equation and write the equation for $\frac{1}{2\pi i}\oint_{\partial\DD_n}\F$, which is just the homogeneous part of the recurrence. The solution is thus a multiple of the eigenvector $y_0$ (noting that $R_n=-i$Res$(g,q_n)$) and Corollary \ref{Cgenrec} completes the proof.
        \end{proof}

\subsection{Concrete calculations: proof of Theorem\,\ref{T4}\,(ii) }\label{cabstract} This follows directly from the following

\begin{Lemma} For $\omega\in(\frac 1m,\frac1{m-1})$ the $0^{\rm th}$ component of the residue of $(I-K)^{-1}\ff$ is

  \begin{equation}
    \label{eq:33}
R_0=\frac{i m\a^m p_0}{2^m\prod_{k<m}(1-\sqrt{1-k\omega})} (1+o_a)
\end{equation}
(with $o_a$ defined in the beginning of \S\ref{PerturbTh}. The other components are of higher order in $\a$.
\end{Lemma}

{\em Proof.}
 Straightforward calculations based on the continued fraction \eqref{eq:24}, \eqref{eq:24O}, and their matching condition  \eqref{eq:root}, \eqref{eq:root}  yield power series in small $\alpha$ for $\yy_0,\, \yy_0^*$ for \eqref{residue}. Note that the residues do not depend on the normalization of these vectors. 
 
 We prove Theorem\,\ref{T4}\,(ii) only for $\omega>1$ (this corresponds to $m=1$ in \eqref{eq:33}). For the other cases the proof is similar, retaining  a sufficient number of components of the vectors $\yy_0,\ \yy_0^*$.

 Choosing the component $-1$ of $\yy_0$, $\yy_{0,-1}=1$, we obtain that $\yy_{0,0}=O(\alpha)$, and $\yy_{0,1}=-\yy_{0,-1}+O(\alpha^2)$. Inductively, $\yy_{0,\pm n}=O(\alpha^{|n|-1})$. Similarly, we chose $\yy^*_{0,0}=1$.
 
More concretely, we let $P$ be the orthogonal projector in $\ell^2$ on the components $-1,0,1$ and, 
to simplify the notation, we omit the factor $(1+o_a)$ in the formulas below. Then
   \begin{equation}
     \label{eq:rr0}
   P\mathbf{y}_0=
                   \left[ \begin {array}{c} -1\\ \noalign{\medskip}{\frac {\a \,
 \left( i\sqrt {\omega+1}-\sqrt {\omega-1} \right) }{2\omega}}
\\ \noalign{\medskip}1
                             \end {array} \right]
                        ,\ \ P\mathbf{y}_0^*=
                        \left[ \begin {array}{c} 
                                 {\frac {\a \, \left( \sqrt {\omega
-1}-i \right) }{2\omega}}\\ \noalign{\medskip}1\\ \noalign{\medskip}{
\frac {-i\a \, \left( \sqrt {\omega+1}+1 \right) }{2\omega}} 
\end {array} \right] ;\ \ 
 P \mathbf{f}_0=
     \left[ \begin {array}{c} {\frac {-i\sqrt {\omega-1}+\sqrt {
\omega+1}}{2 \a }}\\ \noalign{\medskip}{\frac {\a }{\omega}}
\\ \noalign{\medskip}{\frac {i\sqrt {\omega-1}-\sqrt {\omega+1}}{
                           2\a }}\end {array} \right]
\end{equation}
and
\begin{equation}
  \label{eq:S0}
  PS_0P= \left[ \begin {array}{ccc} 0&{\frac {2\,i \left( i\sqrt {\omega-1}
\sqrt {\omega+1}-1 \right) }{{\a }^{3}}}&0\\
\,{\frac { \left( 2\,i\sqrt {\omega-1}+\omega-2 \right) \a }{4{\omega}^{2}\sqrt 
{\omega-1}}}&0&{\frac {-i\a \, \left( 2\,\sqrt {
\omega+1}+\omega+2 \right) }{4{\omega}^{2}\sqrt {\omega+1}}}
\\0&{\frac {-2\,i \left( i\sqrt {\omega-1}\sqrt {
\omega+1}-1 \right) }{{\a }^{3}}}&0\end {array} \right] 
\end{equation}
Using \eqref{residue} we get  $\lambda=\a/2(1+o_a)$. The contribution from the pole close to $p_0$ to the inverse Laplace transform of $\Phi$ is $\tfrac{\a}{2}e^{p_0 t}$. The rest is a straightforward calculation based on \eqref{eq:theta1T}.
\subsection{Proof of the transseries representation}\label{Pfseries}
\subsubsection{Proof of Theorem\,\ref{Theo2}}\label{pf2} We now go back from the discretized quantity $\gg$ to the continuous one, $g$, using  \eqref{notaq}. Then by \eqref{star} and \eqref{Phig} we have
\begin{equation}\label{fiLg}
\phi(t)=(2\pi i)^{-1}\int_{c-i\infty}^{c+i\infty} g(ip) e^{pt} dp
\end{equation}
 We can choose $c=0$ since $g$ is $L^1$. $\gg$ solves \eqref{Kdesigmaop} and,  by Proposition\,\ref{Prop2} and Theorem\,\ref{T-npoles}, it is analytic in $\sigma$ except for an array of poles and of branch points, therefore the same holds for $g$.
We proceed as described in Remark \ref{cuts} and deform the contour of the inverse Laplace transform in \eqref{fiLg} (the Fourier transform of $g$) 
into a sum of Hankel contours. In the process we collect the residues $R_n e^{p_n t}$.  

\z More precisely, we obtain from \eqref{fiLg}
\begin{multline}\label{serphi}
\phi(t)=\frac {1}{2\pi}\int_{ic-\infty}^{ic+\infty} g(q)e^{-iqt}\, dq=i\sum_n{\rm Res}[g(q)e^{-iqt},q_n]\\
+\frac{1}{2\pi i}\sum_n e^{-i\beta_n t}\int_0^\infty \left[ g(-i\tau+\beta_n+0)-g(-i\tau+\beta_n-0)\right]e^{-\tau t}\, d\tau\\
=\sum_nR_ne^{p_nt}+\frac{1}{2\pi i}\sum_n e^{-i\beta_n t}\int_0^\infty \Delta g(-i\tau+\beta_n)e^{-\tau t}\, d\tau
\end{multline}

We only need to check the convergence of the sum of the residues, and of the integrals, which is rather straightforward but for completeness we outline  below. 

The sum of the residues converges factorially fast, by \eqref{eq:43}. We claim that $\sup_{n\in\ZZ,\tau>0}|(n^2+1)g(n\omega\pm 0 -i\tau|<\infty$ ensuring the convergence of the sum of the branch-cut contributions. 

 Fix some $n_0>0$ (a similar argument applies for $n_0<0$) so that $\sup_{n\in\ZZ,\tau>0}|h(n\omega\pm 0 -i\tau)|<1/2$. Note that $g(n_0\omega\pm 0 -i\tau)$ are real-analytic in $\tau$ and vanish in the limit $\tau\to \infty$. Thus $\sup_{\tau>0}|g(n_0\omega\pm 0 -i\tau)|<C$ for some $C$. By analytic continuation, $g(n\omega\pm 0 -i\tau)$ are given by $(I-K_0)^{-1}f$ and in particular are in $\ell^2$ for all $\tau\Ge 0$, and satisfy the recurrence \eqref{Kdesigma}. With $P$ the projection on $\ell^2(n_0+\NN)$ (the $\ell^2$ sequences indexed starting with $n_0$), $P g$  satisfies the equation 
 \begin{equation}
   \label{eq:ctr}
   Pg=PKPg+Pf+E
 \end{equation}
where $E$ is the vector whose only nonzero component is $E_{n_0}=-\alpha h_{n_0}g_{n_0}$. We consider \eqref{eq:ctr} in the space $\ell^\infty_2:=\{x:\|x\|=\sup_{n>n_0} n^2 |x_n| <\infty\}$. Since $\|PK\|<1/2$, $\|f\|<C_1,\|E\|<C$ for some $C_1$, $(I-PK)^{-1}(f+E)$ is the unique  $\ell^\infty_2$ solution of \eqref{eq:ctr}. Since it is obviously in $\ell^2$, it coincides with $g$. 

Finally note that from \eqref{eq:theta1} we obtain that $\theta(t)=-2i\int_t^\infty \phi(s)ds$ (since $\theta(t)\to 0$ as $t\to\infty$). The series for $\theta$ and $\Theta$ are obtained by integration of \eqref{serphi} completing the proof. 

\subsection{Proof of Theorem\,\ref{T4} (iii) and (iv)}\label{PfT434}
Formula \eqref{eq:ThetaSmall a} is obtained from \eqref{eq:341} and the Neumann series for $(I-K_0)^{-1}$, noting that $\Phi=(I-K_0)^{-1}f$.

\subsection{Proof of Theorem\,\ref{T4}\,(iii), (iv): Branch cut contributions}\label{bccon}

For small $\a$, these are most easily found from the Neumann series 
$\gg=(I-\a\tilde{K}_0(\sigma,\omega))^{-1}\ff=\ff+\a\tilde{K}_0\ff+O(\a^2)$
noting that $g_0$ is analytic and $g_n$ for $n\ne\pm1$ has square root branch points in the order $\a^3$ of the expansion (recall that $\ff$ is a multiple of $\a$). Since $g_{\pm 1}=\mp h_0f_0+$analytic$+O(\a^3)$ we obtain, for small $|z|$,
\begin{equation}
g_{\pm 1}(1+z)=\pm\frac{\a^2\omega}{2}\frac 1{\sqrt{z}-i}\frac{\omega}{(1+z)^2-\omega^2} +O(\a^3)+f(z)
\end{equation}
where $f(z)$ is a function analytic at $0$.

The rest follows from Theorem\,\ref{Theo2} using \eqref{serphi} and Theorem\,\ref{T4}\,(ii).

 \section{Appendix}
 
We sketch the argument in the proof of Proposition\,\ref{Prop2} with the notation of this paper.
\begin{proof}
Consider a solution $\gg\in\ell^2$ of \eqref{eq:rec1h}, the homogeneous part of \eqref{eq:rec201}.  Note that $h_0=h_0(\sigma)$ has a pole at $\sigma=0$, hence, from \eqref{eq:not}, $b_1(0)=0$.
Then taking $n=1$ in  the homogeneous part of \eqref{eq:rec201} we see that $g_0(0)=0$. 

Let $g_n h_n=u_n$; clearly $\mathbf{u}\in\ell^2$. Taking now $n=0$ in the recurence, we see that $u_1(0)=u_{-1}(0)$. Reversing now the steps that led to  \eqref{eq:rec201}, we see that the homogeneous equation is equivalent to
\begin{equation}
  \label{eq:hom3}
  h_n^{-1}\,u_n= \a u_{n+1}-\a u_{n-1}
\end{equation}
 where the left side is zero if $n=0,\sigma=0$. 
 Taking scalar product with $\bfu$ in \eqref{eq:hom3} we get
$$\sum_{n\in\ZZ}h_n^{-1}|u_n|^2= \a\sum_{n\in\ZZ}u_{n+1}\overline{u}_n- \a\sum_{n\in\ZZ}u_{n-1}\overline{u}_n=2i \a \Im \sum_{n\in\ZZ}u_{n+1}\overline{u}_n$$
Note that $h_n\in i\RR$ for $n<1$ and $\Re h_n>0$ if $n\ge 1$ and thus, unless all $u_n$ for $n>0$ vanish, the left side above has a positive real part. Then, \eqref{eq:hom3} implies $u_n=0$ for all $n\in\ZZ$. 
\end{proof}
\section{Acknowledgments}
OC  was partially supported by the NSF-DMS grant 1515755 and JLL by the AFOSR grant FA9550-16-1-0037. We thank David Huse for very useful discussions. JLL thanks the Systems Biology division of the Institute for Advanced Study for hospitality during part of this work.


\begin{thebibliography}{99}


\bibitem{tutorial} Bandrauk A.D., Fillion-Gourdeau F., Lorin E., {\em Atoms and molecules in intense laser fields: gauge invariance of theory and models}, J. Phys. B 46 (2013)

\bibitem{BauerD}  Bauer D., {\em Theory of intense laser-matter interaction}, Lecture notes, Univ. of Heidelberg,
(2006)

\bibitem{22} Bauer J.H., \em{Keldysh theory re-examined},  J. Phys. B 49 (2016)


\bibitem{Intro} Cohen-Tannoudji C., Dupont-Roc J., Grynberg G., {\em Photons and Atoms. Introduction to quantum electrodynamics.}, John Wiley and Sons Ltd, 1997

\bibitem{CT2} Cohen-Tannoudji C., Diu B., Laloe F., {\em Quantum mechanics. Volume 2}, Wiley (1991).




\bibitem{OCJLAR} Costin O, Lebowitz J.L., Rokhlenko A. , {\em Exact Results for the Ionization of a Model Quantum System} J. Phys. A: Math. Gen.  33 pp. 1--9 (2000) 

\bibitem{prepa} Costin O, Costin R.D., Lebowitz J.L., Rokhlenko A., {\em  Ionization by an Oscillating Field: Where are the Photons?}, in preparation

\bibitem{CMP} Costin O, Costin R.D., Lebowitz J.L., Rokhlenko A., {\em  Evolution of a model quantum system under time periodic forcing: conditions for complete ionization} Comm. Math. Phys.  221, 1 pp 1-26 (2001).

\bibitem{CLS} Costin, O, Lebowitz, J.L., Stucchio, C, {\em Ionization in a 1-dimensional dipole model.} Rev. Math. Phys. 20 (2008), no. 7, 835-872

\bibitem{cycon} Cycon, H.L., Froese, R.G., Kirsch, W. and Simon, B.: {\em Schr\"odinger Operators.} Springer-Verlag, (1987).
  



\bibitem{Elberfeld} W. Elberfeld and M. Kleber, {\em Tunneling from an ultrathin quantum well
in a strong electrostatic field: A comparison of different methods.} Z. Phys.B-- Condensed Matter 73, 23--32 (1988). 



\bibitem{Linop} Gohberg I., Goldberg S., Kaashoek M. A., {\em Basic Classes of Linear Operators}, Springer Birkh\"auser, (2003)

\bibitem{Immink} Immink, G. K., \em{Asymptotics of Analytic Difference Equations}, Springer, (1984)


\bibitem{37} Karnakov B.M et al {\em Current progress in developing the nonlinear ionization theory of atoms and ions}, Physics-Uspekhi 58 (3), (2015)


\bibitem{28} Karnakov B.M. {\em Nonperturbative generalization of the Fermi golden rule}, JETP Letters 101 (825), (2015)

\bibitem{Krantz} Krantz, S. G.  Handbook of Complex Variables. Boston, MA: Birkh\"auser, p. 159, (1999).


\bibitem{Keldysh} Keldysh L. V. {\em Ionization in the field of a strong electromagnetic wave}, Sov. Phys.Jetp  (1964)  (Engl. transl.),  Zh. Eksp. Teor. Fiz. 47 (1945)








\bibitem{Popruzhenko} Popruzhenko S.V., {\em Keldysh theory of strong field ionization: history, applications, difficulties and
perspectives}, J. Phys. B: Atomic, Molecular and Optical Physics 47 (20) (2014)

\bibitem{Floquet} Gu\'erin S., Monti F., Dupont J-M., Jauslin H. R., {\em On the relation between cavity-dressed states, Floquet states, RWA and semiclassical models}, J. Phys. A, 30 (1997)
  
  \bibitem{Reed-Simon1}
{Reed, M.; Simon, B.}
\newblock {\em {Methods of Modern Mathematical Physics I}}.
\newblock Academic Press, 1980.










\bibitem{Zhang} Zhang P., Lau Y.Y., {\em Ultrafast strong-field photoelectron emission from biased metal surfaces: exact solution to time-dependent Schr\"odinger Equation}, Nature, Scientific Reports (2016)




\bibitem{Zinn} Zinn-Justin J.,  Jentschura U. D., \em{Multi-instantons and exact results I: conjectures, WKB expansions, and instanton interactions}, Annals of Phys, 313 (1) 197-267 (2004)


 \end{thebibliography}
\end{document}